\documentclass[12pt]{article}
\usepackage{amsmath,amssymb,amsthm}
\usepackage{amsfonts,dsfont,mathrsfs}
\usepackage{graphicx}
\usepackage{mathdots}
\usepackage{color}
\usepackage{bm,nccmath}
\PassOptionsToPackage{hyphens}{url}
\usepackage{hyperref}

\newcommand*{\pd}
[2]{\mathchoice{\frac{\partial#1}{\partial#2}}
  {\partial#1/\partial#2}{\partial#1/\partial#2}
  {\partial#1/\partial#2}}
\newcommand*{\fd}
[2]{\mathchoice{\frac{\delta#1}{\delta#2}}
  {\delta#1/\delta#2}{\delta#1/\delta#2}
  {\delta#1/\delta#2}}
\newcommand{\ddx}[1]{\partial_x^{#1}}
\newtheorem{theorem}{Theorem}
\newtheorem*{conjecture*}{Conjecture}
\newtheorem*{algorithm*}{Algorithm}

\newtheorem{definition}[theorem]{Definition}

\newtheorem{remark}[theorem]{Remark}

\begin{document}

\title{WDVV equations and\\
  invariant bi-Hamiltonian formalism}
\author{J. Va\v{s}\'\i\v{c}ek$^{1,2}$, R. Vitolo$^{2,3}$ \\[3mm]
  \small $^1$Mathematical Institute of the Silesian University,
  \\
  \small Silesian University, Opava, Czech Republic
  \\
  \small \texttt{jakub.vasicek@math.slu.cz}
  \\
  \small $^{2}$Department of Mathematics and Physics \textquotedblleft E. De
  Giorgi\textquotedblright ,\\
  \small University of Salento, Lecce, Italy\\
  \small and  $^3$Istituto Nazionale di  Fisica Nucleare -- Sez.\ Lecce\\
  \small \texttt{raffaele.vitolo@unisalento.it} } \date{\itshape
  Respectfully dedicated to the memory of\\[0.5mm]
  Boris Anatolevich Dubrovin (1950--2019)}
\maketitle

\begin{abstract}
  The purpose of the paper is to show that, in low dimensions, the WDVV
  equations are bi-Hamiltonian. The invariance of the bi-Hamiltonian formalism
  is proved for $N=3$. More examples in higher dimensions show that the result
  might hold in general. The invariance group of the bi-Hamiltonian pairs that
  we find for WDVV equations is the group of projective transformations. The
  significance of projective invariance of WDVV equations is discussed in
  detail. The computer algebra programs that were used for calculations
  throughout the paper are provided in a GitHub repository.

  \bigskip


  \bigskip

\end{abstract}

\newpage

\tableofcontents

\section{Introduction}

The Witten--Dijkgraaf--Verlinde--Verlinde (WDVV) equations are an
overdetermined system of Partial Differential Equations (PDEs) that originated
in two-dimensional topological field theory
\cite{DIJKGRAAF199159,WITTEN1990281}. Since then, they have been a central
subject in Theoretical Physics, with applications ranging from supersymmetric
quantum mechanics \cite{Antoniou2019,Galajinsky2016,Kozyrev2018}, topological
quantum field theory \cite{Gomez2021}, string theory
\cite{Belavin2014,Ding2016} and supersymmetric gauge theory \cite{Jockers2019}.

In nearly all the above research subjects, WDVV equations are associated with
integrable systems. The deep connections between these two fields were first
described in the seminal work of B.A. Dubrovin \cite{D96}. Since then, WDVV
equations continue to attract the attention of Mathematicians. Few recent works
show the depth and breadth of their investigation \cite{Basalaev_2021,Cao2021,
  solomon2021relative,StedmanStrachanJMP2021,
  ferapontov2020secondorder,cotti2020degenerate,
  zinger2020real,Shen_2017,Strachan_2017}.

The aim of this paper is to present new results on the Hamiltonian formalism
for WDVV equations and related geometric properties, and eventually stimulate
further research along this direction in the field.

The Hamiltonian formalism for PDEs has been developed for more than 50 years
with the idea to reproduce geometric structures and mathematical results
analogous to those of integrable systems in Hamiltonian mechanics (see,
\emph{e.g.}, \cite{NovikovManakovPitaevskiiZakharov:TS}).

The main difficulty in bringing the Hamiltonian formalism to the WDVV equations
is that the Hamiltonian formalism was developed for evolution PDEs, while
WDVV equations are an overdetermined system of PDEs in a single unknown
function. More precisely, an evolutionary system of PDEs of the form
\begin{equation}
  \label{eq:9}
  u^i_t = f^i(u^j,u^j_x,u^j_{xx},\ldots),\qquad i=1,\ldots,n
\end{equation}
in $n$ unknown functions of two independent variables $u^i=u^i(t,x)$ is said to
be \emph{Hamiltonian} if there exists a linear differential operator
$A=A^{ij\sigma}\partial_{\sigma}$, where
$A^{ij\sigma}=A^{ij\sigma}(u^k,u^k_x,u^k_{xx},\ldots)$ and
$\partial_\sigma=\partial_{x}\circ\cdots\circ\partial_x$ ($\sigma$-times), and
a density $H=\int h\,dx$, where $h=h(u^k,u^k_x,u^k_{xx},\ldots)$ such that
\begin{equation}
  \label{eq:30}
  u^i_t= f^i(u^j,u^j_x,u^j_{xx},\ldots) = A^{ij\sigma}\partial_\sigma\fd{H}{u^j}.
\end{equation}
$H$ is said to be a \emph{Hamiltonian density}.
The operator $A$ is required to define a \emph{Poisson bracket} between
conserved densities $F$, $G$ of the PDE:
\begin{equation}
  \label{eq:35}
  \{F,G\}_A = \int\fd{F}{u^i}A^{ij\sigma}\partial_\sigma \fd{G}{u^j}\,dx.
\end{equation}
The skew-symmetry of the Poisson bracket is equivalent to the
skew-ad\-joint\-ness of the operator: $A^*=-A$, and the Jacobi identity is
equivalent to the vanishing of the \emph{Schouten bracket} of the operator:
$[A,A]=0$ (see, \emph{e.g.} \cite{Dorfman:DSInNEvEq,GelfandDorfman:SBHOp,
  KVV17,magri08:_hamil_poiss}). An operator $A$ fulfilling the above properties
is said to be a \emph{Hamiltonian operator}.

A \emph{bi-Hamiltonian} system of PDEs is just a system of PDEs that is
Hamiltonian with respect to two operators $A_1$, $A_2$ (and the respective
Hamiltonian densities). The operators are required to be \emph{compatible}:
their Schouten bracket vanishes $[A_1,A_2]=0$, or the pencil $A_1+\lambda A_2$
is a Hamiltonian operator for every $\lambda\in\mathbb{R}$. In this case,
Magri's Theorem \cite{Magri:SMInHEq} yields an infinite sequence of commuting
conserved quantities or symmetries, which is usually identified with
\emph{integrability}.

In order to see how to transform the WDVV system into an evolutionary system,
we shall first recall the basic notions. We will follow
\cite{dubrovin06:_encyc_mathem_physic}. The mathematical problem is: in
$\mathbb{R}^N$ find a function $F=F(t^1,\ldots,t^N)$ such
that\label{page:introduction}
\begin{enumerate}
\item
  \begin{sloppypar}$\displaystyle \frac{\partial^3 F}{\partial t^1\partial
      t^\alpha\partial t^\beta} = \eta_{\alpha\beta}$ is a constant symmetric
    nondegenerate matrix;
  \end{sloppypar}
\item
  $c^\gamma_{\alpha\beta} = \eta^{\gamma\epsilon} \displaystyle
  \frac{\partial^3 F}{\partial t^\epsilon\partial t^\alpha\partial t^\beta}$
  are the structure constants of an associative algebra;\label{item:2}
\item $F$ is quasihomogeneous:
  $F(c^{d_1}t^1,\ldots,c^{d_N}t^N) = c^{d_F}F(t^1,\ldots,t^N)$.\label{item:1}
\end{enumerate}
If $e_1$,\dots, $e_N$ is the basis of $\mathbb{R}^N$ then the algebra operation
is $e_\alpha\cdot e_\beta = c^\gamma_{\alpha\beta}(\mathbf{t})e_\gamma$ with
unity $e_1$. The WDVV, or associativity, system of PDEs takes the form
\begin{equation}\label{eq:5}
  \eta ^{\mu \lambda }\frac{\partial ^{3}F}{\partial t^{\lambda }\partial
    t^{\alpha }\partial t^{\beta }}\frac{\partial ^{3}F}{\partial t^{\nu
    }\partial t^{\mu }\partial t^{\gamma }}=\eta ^{\mu \lambda }\frac{\partial
    ^{3}F}{\partial t^{\nu }\partial t^{\alpha }\partial t^{\mu }}\frac{\partial
    ^{3}F}{\partial t^{\lambda }\partial t^{\beta }\partial t^{\gamma }}.
\end{equation}
The unknown of the system is not exactly $F$, as the above requirements
completely specify the functional dependence from $t^1$ (up to a second degree
polynomial, see \cite{dubrovin06:_encyc_mathem_physic}):
\begin{equation}\label{eq:7}
  F=\frac{1}{6}\eta_{11}(t^1)^3 +
  \frac{1}{2}\sum_{k>1}\eta_{1k}t^k(t^1)^2 + \frac{1}{2}\sum_{k,s>1}
  \eta_{sk}t^st^kt^1 + f(t^2,\ldots,t^N).
\end{equation}
This implies that the WDVV system is an overdetermined system in one unknown
function $f$ of $N-1$ independent variables. Just as an example, in the case
$N=3$ we have a single equation on $f=f(t^2,t^3)=f(x,t)$. It was proved in
\cite{D96} that, when $\eta_{11}=0$, the matrix $\eta_{\alpha\beta}$ can be
transformed by a linear change of coordinates that preserves $\pd{}{t^1}$ to
\begin{equation}\label{eq:8}
  \eta_{\alpha\beta}= \delta_{\alpha+\beta,N+1}=
  \begin{pmatrix}
    0 & 0 & 1
    \\
    0 & 1 & 0
    \\
    1 & 0 & 0
  \end{pmatrix},
\end{equation}
and the WDVV equation becomes
\begin{equation}\label{eq:13}
  f_{ttt}=f_{xxt}^2 - f_{xxx}f_{xtt}.
\end{equation}

A technique was developed by O. Mokhov \cite{mokhov95:_sympl_poiss} in order to
rewrite the WDVV equation in the case $N=3$ as a first-order quasilinear
system, or hydrodynamic-type system, of PDEs. Namely, if we introduce
coordinates $a=f_{xxx}$, $b=f_{xxt}$, $c=f_{xtt}$ then for~\eqref{eq:13} we
have the compatibility conditions
\begin{equation}\label{eq:6}
  \left\{
    \begin{array}{l}
      a_t = b_x,\\ b_t = c_x,\\ c_t = (b^2 - ac)_x
    \end{array}
  \right.
\end{equation}
We will say that the above system is a \emph{first-order WDVV system}.  The
above system is of the general conservative first-order quasilinear form
\begin{equation}
  \label{eq:15}
  u^i_t = (V^i(\mathbf{u}))_x = \pd{V^i}{u^j}u^j_x,
\end{equation}
where $u^i=u^i(t,x)$ are field variables, $i=1$, \dots, $n$. The above
representation allowed to find a bi-Hamiltonian formalism for the
equation~\eqref{eq:13} \cite{FGMN97}:
\begin{equation}\label{eq:16}
  u_{t}^{i}=A_{1}^{ij}\fd{H_2}{u^j}=A_{2}^{ij}\fd{H_{1}}{u^j}.
\end{equation}%
with respect to two compatible local Hamiltonian operators $A_{1}$ and $A_{2}$,
with expressions
\begin{subequations}
  \label{eq:11}
  \begin{gather}\label{eq:39}
    A_{1}=%
    \begin{pmatrix}
      -\frac{3}{2}\partial _{x}^{{}} & \frac{1}{2}\partial _{x}^{{}}a &
      \partial
      _{x}^{{}}b \\
      \frac{1}{2}a\partial _{x}^{{}} & \frac{1}{2}(\partial
      _{x}^{{}}b+b\partial
      _{x}^{{}}) & \frac{3}{2}c\partial _{x}^{{}}+c_{x} \\
      b\partial _{x}^{{}} & \frac{3}{2}\partial _{x}^{{}}c-c_{x} &
      (b^{2}-ac)\partial _{x}^{{}}+\partial _{x}^{{}}(b^{2}-ac)%
    \end{pmatrix},
    \\
    \label{eq:40}
    A_{2}=%
    \begin{pmatrix}
      0 & 0 & \partial _{x}^{3} \\
      0 & \partial _{x}^{3} & -\partial _{x}^{2}a\partial _{x} \\
      \partial _{x}^{3} & -\partial _{x}a\partial _{x}^{2} & \partial
      _{x}^{2}b\partial _{x}+\partial _{x}b\partial _{x}^{2}+\partial
      _{x}a\partial _{x}a\partial _{x}%
    \end{pmatrix}.
  \end{gather}%
\end{subequations}
The above Hamiltonian operators are homogeneous with respect to the grading
$\deg \partial _{x}=1$. The Hamiltonian densities are $H_2 = \int c\,dx$ (for
$A_1$) and
$H_1=\int[ -(1/2)a(\partial_x^{-1}b)^2 -
(\partial_x^{-1}b)(\partial_x^{-1}c)]\,dx$ (for $A_2$; this one is nonlocal).
It should be remarked that B. Dubrovin proved that WDVV equations are
integrable by providing a Lax pair for arbitrary $N$ \cite{D96}. Nonetheless,
knowing the bi-Hamiltonian formalism for a system of PDEs is an additional
source of information. In the WDVV case, it will be shown here that the
additional information is provided by the invariance properties of the
bi-Hamiltonian structure.

First-order homogeneous Hamiltonian operators (HHOs) were introduced in
\cite{DN83}. They have the form
\begin{equation}
  \label{eq:17}
  A_1^{ij} = g^{ij}\partial_x + \Gamma^{ij}_k u^k_x,
\end{equation}
where $g^{ij}=g^{ij}(\mathbf{u})$ transforms as a symmetric contravariant
tensor (we will always assume that $\det(g^{ij})\neq 0$) whose inverse $g_{ij}$
is a flat pseudo-Riemannian metric with Christoffel symbols
$\Gamma^{i}_{jk} = - g_{js}\Gamma^{si}_k$.

It should be stressed that solutions of WDVV equations yield Frobenius
manifolds, or integrable hierarchies of PDEs defined by bi-Hamiltonian pairs of
first-order HHOs (see, \emph{e.g.},
\cite{dubrovin01:_normal_pdes_froben_gromov_witten,dubrovin98:_flat_froben,D96}).
However, the above quasilinear system of first-order PDEs \eqref{eq:6} is
\emph{exceptional} with respect to the theory of Frobenius manifolds as it is
bi-Hamiltonian with respect to a pair of a first-order HHO and a third-order
HHO. Higher order homogeneous Hamiltonian operators were introduced in
\cite{DubrovinNovikov:PBHT}, and have a considerably more complicated structure
than \eqref{eq:17}. Third-order HHOs can always be transformed to the canonical
form
\begin{equation}
  \label{eq:18}
  A_2 = \partial_x(h^{ij}\partial_x + c^{ij}_k u^k_x)\partial_x,
\end{equation}
(again, we require that the leading coefficient is non-degenerate:
$\det(h^{ij})\neq 0)$
\cite{doyle93:_differ_poiss,potemin91:PhDt,potemin97:_poiss,balandin01:_poiss}
which is invariant with respect to the action of projective reciprocal
transformations \cite{FPV14,FPV16} (see Section \ref{sec:preliminaries} for
more details).

In further papers it was shown that a bi-Hamiltonian formulation as above
exists for a different choice of the matrix $\eta_{ij}$ (in the case $N=3$)
\cite{kalayci97:_bi_hamil_wdvv} or after the exchange of $t$ and $x$ in
\eqref{eq:13} \cite{kalayci98:_alter_hamil_wdvv}, and much more recently, in
the case $N=4$ for $\eta^{(1)}$ \eqref{eq:14} \cite{PV15}. In a new interesting
paper \cite{mokhov18:_class_hamil}, the classification of $N=3$ WDVV equations
admitting a Hamiltonian formalism with a \emph{local} first-order HHO as in
\eqref{eq:11} was given.

It was natural to try to prove that WDVV equations admit a bi-Hamilto\-nian
formulation by means of a compatible pair of a first-order HHO and a
third-order HHO for \emph{any choice of $\eta$}. The idea for the proof is
first to prove that the invariance group of WDVV equations, \emph{i.e.}  linear
transformations in the space $(t^1,\ldots, t^N)$ that leave $\pd{}{t^1}$
invariant \cite{dubrovin06:_encyc_mathem_physic}, do not change the form of a
bi-Hamiltonian pair as above. Then, it is enough to prove the statement only on
normal forms with respect to the invariance group.

Indeed, using the invariance group of the problem, one can reduce the matrix
$\eta_{ij}$ to two canonical forms if the quasihomogeneity weights are
distinct:
\begin{description}
\item[$\eta_{11}=0$:] the canonical form is
  $\eta^{(1)}=(\eta^{(1)}_{\alpha\beta})$ with
  \begin{equation}\label{eq:14}
    \eta^{(1)}_{\alpha\beta}= \delta_{\alpha+\beta,N+1}=
    \begin{pmatrix}
      0 & & 1
      \\
      & \iddots &
      \\
      1 & & 0
    \end{pmatrix}
  \end{equation}
  where
  $F = \frac{1}{2}(t^1)^2t^N + \frac{1}{2}t^1\sum_{\alpha=2}^{N-1}t^\alpha
  t^{N-\alpha+1} + f(t^2,\ldots,t^N)$;
\item[$\eta_{11}\neq 0$:] this case can only happen if $d_F=3$, and the
  canonical form is $\eta^{(2)}=(\eta^{(2)}_{\alpha\beta})$ with
  
  \begin{equation}\label{eq:10}
    \eta^{(2)}_{\alpha\beta}= \delta_{\alpha+\beta,N+1}=
    \begin{pmatrix}
      \mu & & 1
      \\
      & \iddots &
      \\
      1 & & 0
    \end{pmatrix}
  \end{equation}
  where $\mu\neq 0$ and
  $F = \frac{\mu}{6}(t^1)^3 + \frac{1}{2}t^1\sum_{\alpha=2}^{N}(t^{\alpha})^2 +
  f(t^2,\ldots,t^N)$.
\end{description}
If one drops the quasihomogeneity request on $F$, a smaller group of linear
transformations can be used to show that if $N=3$ then there are $4$ distinct
canonical forms \cite{mokhov18:_class_hamil}.  The results in
\cite{mokhov18:_class_hamil} also suggested that we should consider a wider
class of first-order HHOs, namely the non-local Ferapontov operators (see
\emph{e.g.} \cite{F95:_nl_ho}). It should be stressed that the necessary
theoretical background and software for computations with nonlocal operators
was not available until recently \cite{m.20:_weakl_poiss,CLV19}.

That led to the first part of the results in this paper. Namely, in the case
$N=3$ we have:
\begin{itemize}
\item for any choice of the matrix $\eta$ the first-order WDVV systems
  admit a third-order homogeneous Hamiltonian operator in canonical form
  \begin{equation}\label{eq:122}
    A_2 = \partial_x(h^{ij}\partial_x + c^{ij}_ku^k_x)\partial_x;
  \end{equation}
\item the first-order WDVV systems defined by matrices $\eta$ in the orbit of
  $\eta^{(1)}$ admit a first-order \emph{local} homogeneous
  Hamiltonian operator of the type
  \begin{equation}
    \label{eq:171}
    A_1^{ij} = g^{ij}\partial_x + \Gamma^{ij}_k u^k_x;
  \end{equation}
\item  the first-order WDVV systems defined by matrices $\eta$ in the orbit of
  $\eta^{(2)}$ admit a first-order \emph{non-local} homogeneous
  Hamiltonian operator of Ferapontov type
  \begin{multline}
    \label{eq:37}
    A_1^{ij} = g^{ij}\ddx{} + \Gamma^{ij}_{k}u^k_x + \alpha V^i_qu^q_x\ddx{-1}
    V^j_pu^p_x
    \\
    +\beta\left( V^i_qu^q_x\ddx{-1}u^j_x + u^i_x\ddx{-1}V^j_qu^q_x \right)
    +\gamma u^i_x\ddx{-1}u^j_x,
  \end{multline}
  where $V^i_j=\pd{V^i}{u^j}$ is the matrix of velocities of the first-order
  WDVV system \eqref{eq:15} and $\alpha$, $\beta$, $\gamma$ are three
  constants;
\item finally, first-order WDVV systems are bi-Hamiltonian: the Schouten
  bracket of the two operators vanishes, $[A_1,A_2]=0$, or the pencil
  $A_1+\lambda A_2$ is a Hamiltonian operator for every $\lambda\in\mathbb{R}$.
\item We also realized that the quasihomogeneity of $F$ in the assumptions on
  the WDVV problem (item 3 on page \pageref{item:1}) can be dropped without
  changing all our results in the above items. Indeed, in a recent paper
  \cite{mokhov18:_class_hamil} Mokhov and Pavlenko classified the WDVV
  equations without the requirement of quasihomogeneity of the solutions, and
  obtained $4$ canonical forms. For all of them we recover the bi-Hamiltonian
  pair, as we will show in Section~\ref{sec:o.i.-mokhov-n.a}.
\end{itemize}

The above results imply that in the case $N=3$ the WDVV quasilinear
first-order systems are \emph{linearly degenerate, non diagonalizable and in
  the Temple class}, as it follows from the main results in
\cite{FPV17:_system_cl}. Indeed, the presence of third-order operators yields
many interesting properties of the underlying first-order quasilinear system of
PDEs.

In higher dimensions proving a general invariance theorem is more difficult,
and will be considered in the future. However, if $N=4$ and $\eta=\eta^{(1)}$
it was already known that the first-order WDVV system had a first-order local
HHO \cite{ferapontov96:_hamil} and a compatible third-order HHO \cite{PV15}. In
this paper we prove that
\begin{itemize}
\item if $N=4$ and $\eta=\eta^{(2)}$ then the first-order WDVV system admits a
  third-order HHO of the form~\eqref{eq:18};
\item if $N=5$ and $\eta=\eta^{(1)}$ or $\eta=\eta^{(2)}$ (with $\mu=1$, as the
  case with an arbitrary $\mu\neq 0$ was beyond the capabilities of our
  servers), then the first-order WDVV systems admit a third-order HHO of the
  form~\eqref{eq:18}.
\end{itemize}

We did not try to find the first-order operator in the case $N=4$ and
$\eta=\eta^{(2)}$ or when $N=5$: indeed, the results in
\cite{bogoyavlenskij96:_neces_hamil} that we used to find the first-order
operators in the case $N=3$ do not hold when $N\geq 4$.

Even if invariance is fully stated only in the case $N=3$, the presence of
first-order and third-order HHOs in higher dimensions is enough to support the
following conjecture.

\begin{conjecture*} The WDVV equations in the form of quasilinear systems of
  first-order PDEs are bi-Hamiltonian with respect to a pair of a third-order
  HHO in canonical form~\eqref{eq:18} and a first-order HHO, which can either
  be local \eqref{eq:17} in the case $\eta=\eta^{(1)}$, or nonlocal of general
  Ferapontov type
  \begin{equation}
    \label{eq:371}
    A_1^{ij} = g^{ij}(\mathbf{u})\ddx{} + \Gamma^{ij}_{k}(\mathbf{u})u^k_x +
    \sum_\alpha c^{\alpha\beta} w^i_{\alpha k}(\mathbf{u})u^k_x
    \partial^{-1}_x w^j_{\beta h}(\mathbf{u}) u^h_x,
  \end{equation}
  where $c^{\alpha\beta}$ is a constant symmetric matrix, in the case
  $\eta=\eta^{(2)}$.
\end{conjecture*}

There are interesting implications of the conjecture. Indeed, it was proved in
\cite{FPV14,FPV17:_system_cl} that third-order HHOs can be regarded as
distinguished projective varieties, namely, \emph{quadratic line
  complexes}. They determine families of varieties, \emph{linear line
  congruences}, that correspond to first-order quasilinear systems of PDEs
(first-order WDVV being one such systems in all known cases).

As the above correspondence between integrable systems and projective varieties
is non-standard, it is instructive to show it in the simplest example of WDVV
equation \eqref{eq:13}. We will use the form~\eqref{eq:6}. The leading
coefficient $\bar{g}=(g^{ij})$ of $A_2$ in \eqref{eq:40}:
\begin{equation}
  \label{eq:41}
  (g^{ij})=
  \begin{pmatrix}
    0 & 0 & 1\\ 0 & 1 & -a\\ 1 & -a & 2b + a^2
  \end{pmatrix}
\end{equation}
has the inverse matrix $g=\bar{g}^{(-1)}$ that is a Monge metric. It can be
written in the form
\begin{equation}
  \label{eq:42}
  g = -2bda^2 + 2a \,da\,db + 2\,da\,dc +db^2.
\end{equation}
Now, it is easy to realize that the above metric represents the equation of a
quadratic line complex in the space of lines of the projective space
$\mathbb{P}^3$. Indeed, considering two infinitesimally close points in a
$4$-dimensional projective space $P=[u^0,u^1,u^2,u^3]$ and
$P+dP=(u^0+du^0,u^1+du^1,u^2+du^2,u^3+du^3)$, the minors of the matrix
$(P,P+dP)$ turn out to be of the form $p^{ij}=u^idu^j - u^jdu^i$ (note that
$u^0=1$ and $du^0=0$ when passing to the affine coordinates $a=u^1$, $b=u^2$,
$c=u^3$). Such forms are the \emph{Lie form of Pl\"ucker coordinates}. We
recall that the Pl\"ucker coordinates characterize lines in $\mathbb{P}^3$
modulo the further Pl\"ucker relations
$p^{ij}p^{hk}+p^{ih}p^{kj} + p^{ik}p^{jh}=0$. We can rewrite $g$ as
\begin{equation}
  \label{eq:43}
  g = 2(a\,db-b\,da)\,da + 2 da\,dc + db^2.
\end{equation}
The above metric turns out to be a quadratic expression in the Lie form of
Pl\"ucker coordinates; the corresponding quadratic line complex is given by the
system
\begin{equation}
  \label{eq:44}
  2p^{12}p^{01} + 2p^{01}p^{03} + (p^{02})^2=0, \qquad
  p^{01}p^{23}+p^{02}p^{31} + p^{03}p^{12}=0.
\end{equation}
The line congruence corresponding to the system~\eqref{eq:6} is an
$n$-parameter family of lines in $\mathbb{P}^{n+1}$. In homogeneous coordinates
$[y^1,\ldots,y^{n+2}]$ it has the general form $y^i=u^iy^{n+1}+V^iy^{n+2}$
\cite{agafonov01:_system_templ,agafonov96:_system}. The lines of the congruence
pass through the points $y^i=u^i$, $y^{n+1}=1$, $y^{n+2}=0$ and $y^i=V^i$,
$y^{n+1}=0$, $y^{n+2}=1$, respectively. The corresponding Pl\"ucker coordinates
are the minors of the matrix
\begin{equation}
  \label{eq:2}
  \begin{pmatrix}
    u^1 & \cdots & u^n & 1 & 0\\
    V^1 & \cdots & V^n & 1 & 0
  \end{pmatrix}
\end{equation}
The line congruence is linear if there are $n$ linear relations between the
Pl\"ucker coordinates; these are $n$ linear line complexes. As it was proved in
\cite{FPV17:_system_cl}, every first-order quasilinear system of PDEs admitting
a third-order HHO is associated with a linear line congruence. In the WDVV
example~\eqref{eq:6} the linear line congruence takes the form
\begin{equation}
  \label{eq:45}
  y^1 = ay^{4} + by^{5},\quad y^2 = by^{4} + cy^{5},\quad
  y^3 = cy^{4} + (b^2-ac)y^{5}
\end{equation}
 
The above constructions have a general validity: each time that a third-order
HHO in the form \eqref{eq:18} is found for a system of conservation laws (as we
will do many times in the paper) then one can construct the corresponding
algebraic varieties.

Then, it turns out that the bi-Hamiltonian pairs and the systems of first-order
PDEs are invariant with respect to \emph{projective reciprocal
  transformations}. These are non-local (or non-holonomic) transformations of
the independent variables of the form
\begin{equation}
  \label{recip}
  \begin{array}{c}
    d\tilde x=(a_iu^i+a)dx+(a_iV^i+b)dt, \\
    d\tilde t=(b_iu^i+c)dx+(b_iV^i+d)dt,
  \end{array}
\end{equation}
which, together with the affine transformation of dependent variables, generate
the projective action on the linear line congruence and the quadratic line
complex. More in detail, the above transformation can be factorized in a
sequence of transformations $R_1 \circ E\circ R_2$ where $E$ is just the
exchange of $t$ and $x$ and $R_i$ are transformations of the form
\begin{equation}
  \label{eq:36}
   d\tilde x=(a_iu^i+a)dx+(a_iV^i+b)dt,\qquad d\tilde{t}=dt,
\end{equation}
where the dependent variables undergo a projective transformation
$\tilde{u}^i=(A^i_j u^i + A^i_0)/(a_iu^i+a)$ (see \cite{FPV17:_system_cl} for
the definition). Hence, the whole bi-Hamiltonian WDVV hierarchy becomes a
projective-geometric object (in known cases). The above conjecture can
be rephrased as follows:

\begin{conjecture*}
  To every WDVV system there are associated a quadratic line complex and a
  linear line congruence.
\end{conjecture*}

This fact might have an (at the moment unpredictable) impact in the
applications of WDVV equations. The projective group, in its realization as a
group of distinguished reciprocal transformations, is larger than the
invariance group of WDVV equations (which is the group of linear
transformations leaving $\pd{}{t^1}$ invariant, see
Section~\ref{sec:invar-wdvv-equat}) and its implications in the search for
solutions of WDVV equations is still to be understood. Consequences in
projective and enumerative geometry are not unlikely. See
Section~\ref{sec:projectivegeom} for details.

The paper is structured as follows. Section \ref{sec:preliminaries} describes
the pre-requisites on homogeneous Hamiltonian operators. In
Section~\ref{sec:invar-wdvv-equat} the invariance of bi-Hamiltonian pairs with
respect to invariance transformations of WDVV equations is proved. In
Section~\ref{sec:bi-hamilt-form} bi-Hamiltonian formalism for all normal forms
of WDVV equations in the case $N=3$ is
provided. Section~\ref{sec:wdvv-equations-n=4} considers the problem of finding
Hamiltonian structures for WDVV in higher dimensions. The concluding
Section~\ref{sec:projectivegeom} discusses the projective-geometric aspects of
the results obtained so far.

The calculations have been done by means of computer algebra systems. In
particular, Schouten brackets involving nonlocal operators have been calculated
by the Reduce package CDE, and checked by the Maple package
\texttt{jacobi.mpl}; both packages are described in
\cite{m.20:_weakl_poiss,vitolo17:_hamil_pdes} (see also \cite{KVV17}).  Further
calculations have been done in Reduce (finding first-order nonlocal operators
when $N=3$, and finding third-order operators in the cases $N=3$, $N=4$) and in
Maple, also using the package \texttt{Jets} \cite{BMJets} (finding
third-order operators in the case $N=5$). The programs are available at a
\texttt{GitHub} repository~\cite{vasicek21:_wdvv_hamil}.

\section{Preliminaries: Hamiltonian operators}
\label{sec:preliminaries}

In this paper we will look for Hamiltonian operators for quasilinear systems of
first-order PDEs that are generated by WDVV equations. Known examples suggest
that these can be homogeneous operators of first and third order. Let us
describe such classes more in detail.

First order local homogeneous operators \eqref{eq:17} have already been
described in the Introduction. We will need their nonlocal generalization
\eqref{eq:371}. Such operators were introduced and studied by Ferapontov (see
\cite{F95:_nl_ho}). We will always assume that the leading coefficient is a
non-degenerate matrix: $\det(g^{ij})\neq 0$ (we set
$(g_{ij}) = (g^{ij})^{-1}$). It is well known that the Hamiltonian property is
equivalent to the following conditions: the symmetry of $g^{ij}$, the fact that
$\Gamma^j_{ik} = - g_{ip}\Gamma^{pj}_k$ are the Christoffel symbols of $g_{ij}$
(interpreted as a pseudo-Riemannian metric), and the identities:
\begin{subequations}\label{eq:12}
  \begin{gather}\label{eq:64}
    g^{ik}w^j_{\alpha k} = g^{jk}w^i_{\alpha k}, \\ \label{eq:65} \nabla_k
    w^i_{\alpha j} = \nabla_j w^i_{\alpha k},
    \\
    [w_\alpha,w_\beta] = 0, \\ \label{eq:63} R^{ij}_{kl} =
      c^{\alpha\beta}\Big( w^i_{\alpha k}w^j_{\beta l} - w^j_{\alpha
        k}w^i_{\beta l}\Big).
  \end{gather}
\end{subequations}
Here, $\nabla$ is the Levi-Civita connection of $g_{ij}$,
$R^{ij}_{kl}=g^{is}R^j_{skl}$ (we follow the sign conventions of
\cite{dubrovin98:_flat_froben}), the bracket $[w_\alpha,w_\beta]$ is the usual
commutator of the matrices $w_\alpha = w^i_{\alpha k}$ and
$w_\beta=w^i_{\beta k}$. If the operator is local, then the conditions reduce
to those of the local operators~\eqref{eq:171}.

Third-order homogeneous Hamiltonian operators are much more complicated in
general. However, in the canonical form \eqref{eq:18} (again, we will assume
that the leading coefficient is non-degenerate, $\det(h^{ij})\neq 0$, and we set
$h_{ij}=(h^{ij})^{-1}$) the Hamiltonian property of $A_2$ implies that
$c^{ij}_k$ being given by
\begin{equation}\label{eq:19}
  c_{skm}=\frac{1}{3}(h_{sm,k} - h_{sk,m}),
\end{equation}
where $c_{ijk}=h_{iq}h_{jp}c_{k}^{pq}$, so that the leading coefficient
determines the operator \cite{FPV14}. The Jacobi property further implies that
\cite{FPV14}
\begin{align}
  & h_{mk,s} + h_{ks,m} + h_{ms,k}=0, \label{eq:20} \\
  & c_{msk,l}= - h^{pq}c_{pml}c_{qsk}. \label{eq:21}
\end{align}%
The equation~\eqref{eq:20} is equivalent to the fact that $h_{ij}$ is the Monge
form, or Monge metric, of a quadratic line complex, a distinguished family of
projective varieties. The projective properties of third-order HHOs will be
discussed in Section~\ref{sec:projectivegeom}.

It is important to remark that $h_{ij}$ turn out to be second degree
polynomials with respect to the field variables, under the further algebraic
constraints \eqref{eq:20}.  A complete classification of operators in the
form~\eqref{eq:18} is given in \cite{FPV16,FPV14} for a number of components
$n\leq 4$. The classification uses the projective invariance of third-order
homogeneous operators with respect to reciprocal transformations of the
form~\eqref{eq:36}.

We need a way to find Hamiltonian operators for WDVV systems. This is provided
by the theory of differential coverings
\cite{KerstenKrasilshchikVerbovetsky:HOpC}. In our particular case, it is known
\cite{tsarev85:_poiss_hamil} (but see also \cite{vergallo20:_homog_hamil}) that
a necessary condition for $A_1$ (both in the local and non-local case) to be
the Hamiltonian operator of a quasilinear system of first-order PDEs
\eqref{eq:15} is:
\begin{equation} \label{eq:53} g^{ik}V^j_k = g^{jk}V^i_k,\quad \nabla^iV^j_k =
  \nabla^j V^i_k,
\end{equation}
where $V^i_k = \pd{V^i}{u^k}$ and $\nabla^j$ is the Levi-Civita connection of
$g_{ij}$.  In the case of third-order operators it was recently found
\cite{FPV17:_system_cl} that the compatibility conditions between a third-order
Hamiltonian operators and a quasilinear system of first-order conservation laws
\eqref{eq:15} are:
\begin{subequations}
  \label{V}
  \begin{align}\label{e1}
    & h_{im}V^{m}_{j}=h_{jm}V^m_{i},\\
    \label{e3}
    & c_{mkl}V^m_{i}+c_{mik}V^m_{l}+c_{mli}V^m_{k}=0,\\
    \label{e2}
    &h_{ks}V^k_{ij}=c_{smj}V^m_{i} + c_{smi}V^{m}_{j}.
  \end{align}
\end{subequations}
It is interesting to observe that the conditions \eqref{eq:53} might be
relatively difficult to solve for the metric $g^{ij}$, while the system of
compatibility conditions for $h^{ij}$, being expressed in lower indices, is
indeed a linear algebraic system with respect to the coefficients of the second
degree polynomials $h_{ij}$, which is easy to solve.

\section{WDVV equations as systems of
  \texorpdfstring{\\}{} conservation laws}
\label{sec:invar-wdvv-equat}

As we have seen in the Introduction, the WDVV equations~\eqref{eq:5} in the
unknown function $f=f(t^2,\ldots,t^N)$~\eqref{eq:7} are presented as a system
of conservation laws~\eqref{eq:15} by introducing new dependent variables.  The
general algorithm for the transformation was first given in
\cite{mokhov95:_sympl_poiss} in the case $N=3$ and then generalized for
arbitrary values of $N$ in \cite{ferapontov96:_hamil}, the details are exposed
below.
\begin{algorithm*}
  \indent\par
  \begin{enumerate}
  \item Choose one distinguished independent variable $t^i$, $i>1$ (for example
    $t^2$), and all third-order derivatives of $f$ that contain at least one
    instance of $t^2$; call them $u^1 = f_{t^2t^2t^2}$, $u^2 = f_{t^2t^2t^3}$,
    \dots, $u^n = f_{t^2t^Nt^N}$. These are new dependent variables, and
    $n=N(N-1)/2$.
  \item Choose another independent variable $t^j\neq t^2$, $j>1$ (for example
    $t^3$), and, for any $u^i$, find $u^i_{t^3}$ as the $t^2$-derivative of an
    expression $V^i$:
    \begin{equation}
      u^i_{t^3} = V^i(\mathbf{u})_{t^2}.\label{eq:31}
    \end{equation}
    There are two possibilities:
    \begin{enumerate}
    \item either $V^i(\mathbf{u})$ is one of the coordinates $u^k$, with
      $k\neq i$;
    \item $V^i$ is a third-order derivative of $f$ which is not one of the
      $u^k$.  In this case, $V^i$ must be expressed by means of one of the
      equations of the WDVV system. This is always possible due to the
      structure of the WDVV system.
    \end{enumerate}
  \end{enumerate}
\end{algorithm*}

There are equations in the WDVV system which depend on variables that are not
$t^2$ or $t^3$ derivatives; such equations shall be discarded in the above
construction. However, changing the two distinguished independent variables,
and using other equations in the WDVV system, one obtains $N-2$ distinct
commuting systems of conservation laws with the same structure as above.

In general, first-order WDVV-systems for a fixed (but arbitrary) $N$ and a
fixed choice of $\eta$ are provided by the Algorithm as $N-2$ \emph{commuting}
two-dimensional quasilinear first-order systems of PDEs with $n=N(N-2)/2$
components \cite{ferapontov96:_hamil}.

\begin{definition}
  We say that a quasilinear first-order system of conservation
  laws~\eqref{eq:31} where $(u^i)$ are third-order derivatives of $f$ and the
  equations are compatibility conditions for a WDVV system to be a
  \emph{first-order WDVV system}.
\end{definition}

The purpose of this paper is to show that first-order WDVV systems admit a
third-order HHO in canonical form \eqref{eq:18} and a compatible first-order
local or non-local operator of the type~\eqref{eq:37} for low dimension $N$.

In particular, in the case $N=3$ we will be able to prove that \emph{all}
first-order WDVV systems, \emph{i.e.} for any choice of matrix $\eta$, admit a
bi-Hamiltonian pair as stated above. To this aim, the strategy is:
\begin{enumerate}
\item prove that the invariance group of the WDVV equations do not affect the
  `form' of a bi-Hamiltonian pair as above;
\item prove that there is a bi-Hamiltonian pair as above for each canonical
  form of the WDVV equations, \emph{i.e.}, the canonical form of the matrix
  $\eta_{ij}$.
\end{enumerate}

The invariance group of the WDVV equations with the quasihomogeneity constraint
is the group of linear transformations that preserve the direction of
$\pd{}{t^1}$:
\begin{equation}
  \label{eq:32}
  \tilde{t}^\alpha = P^\alpha_\beta t^\beta + Q^\alpha,\qquad
  \det(P^\alpha_\beta)\neq 0,\quad P^\alpha_1 = \delta^\alpha_1
\end{equation}
\cite{dubrovin06:_encyc_mathem_physic}. We have the transformation rules
\begin{equation}
  \label{eq:33}
  \pd{}{t^\alpha} = P^\beta_\alpha\pd{}{\tilde{t}^\beta},\qquad
  d{\tilde{t}^\alpha} = P^\alpha_\beta dt^\beta,
\end{equation}
which imply that the equation~\eqref{eq:5} is transformed into the same
equation with respect to the new coordinates $(\tilde{t}^\alpha)$ (of course,
one should change coordinates in $F$ and $\eta_{ij}$).

\begin{theorem}\label{th:invariance} Let $N=3$, and suppose that a WDVV system
  in first-order form $u^i_t=(V^i(\mathbf{u}))_x$ is bi-Hamiltonian with
  respect to a pair of compatible Hamiltonian operators $A_1$, $A_2$, where
  $A_1$ is a nonlocal first-order HHO~\eqref{eq:37} and $A_2$ is a local
  third-order HHO~\eqref{eq:18}.

  Then, the coordinate change~\eqref{eq:32} does not change the form of the
  bi-Hamiltonian pair $A_1$, $A_2$.
\end{theorem}
\begin{proof}
  The matrix $P=(P^\alpha_\beta)$ of the change of coordinates can be
  factorized as
  \begin{gather}
    \label{eq:34}
    P= T_1\cdot T_2,\qquad \text{where}
    \\
    \label{eq:56}
    P=
    \begin{pmatrix}
      1 & P^1_2 & P^1_3
      \\
      0 & P^2_2 & P^2_3
      \\
      0 & P^3_2 & P^3_3
    \end{pmatrix},\quad T_1=
    \begin{pmatrix}
      1 & 0 & 0
      \\
      0 & P^2_2 & P^2_3
      \\
      0 & P^3_4 & P^3_3
    \end{pmatrix},\quad T_2=
    \begin{pmatrix}
      1 & P^1_2 & P^1_3
      \\
      0 & 1 & 0
      \\
      0 & 0 & 1
    \end{pmatrix}.
  \end{gather}
  The matrix $T_1$ can be further factorized as
  \begin{gather}\label{eq:52}
    T_1 = R_1\cdot E\cdot R_2,\qquad\text{where}
    \\
    \label{eq:55}
    R_1 =
    \begin{pmatrix}
      1 & 0 & 0
      \\
      0 & \alpha & \beta
      \\
      0 & 0 & 1
    \end{pmatrix},\quad E=
    \begin{pmatrix}
      1 & 0 & 0
      \\
      0 & 0 & 1
      \\
      0 & 1 & 0
    \end{pmatrix},\quad R_2=
    \begin{pmatrix}
      1 & 0 & 0
      \\
      0 & \gamma & \delta
      \\
      0 & 0 & 1
    \end{pmatrix}.
  \end{gather}
  when $P^3_2\neq 0$ (when it is zero no factorization is needed).

  Let us now choose $t=t^3$ and $x=t^2$. When $N=3$ there is only one WDVV
  equation, from which the above Algorithm yields the following system:
  \begin{equation}
    \label{eq:50}
    \begin{split}
      & u^1_t = u^2_x,\\
      & u^2_t = u^3_x,\\
      & u^3_t = \phi(\mathbf{u})_x,
    \end{split}
  \end{equation}
  where $\phi$ is a rational function of the field variables.

  The transformation $T_2$ does not change third-order derivatives of $f$,
  hence it does not affect the system~\eqref{eq:50}.

  The transformation $R_1$ (equivalently, $R_2$) has the effect of a reciprocal
  transformation that preserves the coordinate $t$, namely
  \begin{equation}
    \label{eq:47}
    d\tilde{x} = P^2_2 dx + P^2_3 dt = \alpha dx + \beta dt,
    \qquad d\tilde{t} = dt.
  \end{equation}
  Such transformations are proved to preserve the canonical form of a
  third-order HHO \cite{FPV14} and the locality (or the non-local form) of a
  first-order HHO \cite{pavlov95:_conser_hamil,ferapontov03:_recip_hamil}.  At
  the same time, the third-order derivatives $u^1=f_{xxx}$, $u^2=f_{xxt}$,
  $u^3 = f_{xtt}$ undergo the affine transformation:
  \begin{equation}
    \label{eq:48}
    \begin{split}
      &u^1 = \alpha^3 \tilde{u}^1,\\
      &u^2 = \alpha^2\beta\tilde{u}^1 + \alpha^2\tilde{u}^2,\\
      &u^3 = \alpha\beta^2\tilde{u}^1 + 2\alpha\beta\tilde{u}^2 +
      \alpha\tilde{u}^3,\\
    \end{split}
  \end{equation}
  where $\tilde{u}^1=f_{\tilde{x}\tilde{x}\tilde{x}}$,
  $\tilde{u}^2=f_{\tilde{x}\tilde{x}\tilde{t}}$,
  $\tilde{u}^3 = f_{\tilde{x}\tilde{t}\tilde{t}}$. Again,that does not modify the
  structure of the bi-Hamiltonian pair.

  The transformation $E$ is just an exchange of the independent variables. It
  preserves both the canonical form of the third-order HHO
  \cite{FPV17:_system_cl} and the locality (or the non-local form) of the
  first-order HHO. This completes the proof.
\end{proof}

\begin{remark}
  In the case $N=3$ after a change of coordinates of the type $R_1$ we obtain a
  new WDVV equation. Hence, we can construct a new quasilinear first-order WDVV
  system:
  \begin{equation}
    \label{eq:57}
    \bar{u}^i_{\bar{t}}=(\bar{V}^i(\bar{\mathbf{u}}))_{\bar{x}}
  \end{equation}
  using the above Algorithm. However, it can be proved that in general there
  does not exist an affine transformation that brings the new system into the
  system
  $\tilde{u}^i_{\bar{t}}=(\tilde{V}^i(\tilde{\mathbf{u}}))_{\tilde{x}}$. Later,
  the relation between the two systems will be clarified.
\end{remark}

\begin{remark}
  In the case $N=3$ the transformation $E$ brings the system~\eqref{eq:50} into
  the system
  \begin{equation}
    \label{eq:51}
    \tilde{u}^1_{\tilde{t}} = \phi(\tilde{\mathbf{u}})_{\tilde{x}},\quad
    \tilde{u}^2_{\tilde{t}} =\tilde{u}^1_{\tilde{x}}\quad,
    \tilde{u}^3_{\tilde{t}} = \tilde{u}^2_{\tilde{x}};
  \end{equation}
  interchanging $\tilde{u}^1$ and $\tilde{u}^3$ bring the system in the same
  form as~\eqref{eq:50}.
\end{remark}

The case $N=4$ cannot be treated in the above way. Indeed, the invariance
transformation mix the independent variables, and the commuting systems are
transformed in a more complicated way. However, we will be able to show that
third-order HHO are present for first-order WDVV systems in canonical forms of
WDVV equations when $N=4$ and $N=5$.

\section{Bi-Hamiltonian formalism for WDVV equations,
  \texorpdfstring{$N=3$}{N=3}}
\label{sec:bi-hamilt-form}

Having Theorem~\ref{th:invariance} at hand, we can investigate $N=3$ WDVV
systems and look for bi-Hamiltonian pairs in their representation as
quasilinear systems of first-order PDEs. If we prove that if the WDVV system
obtained by each canonical form of the matrix $\eta_{ij}$ is endowed with a
bi-Hamiltonian pair of the type that we discussed in the Introduction, then
this will be true for an arbitrary matrix $\eta_{ij}$.

We will also discuss bi-Hamiltonian pairs for the canonical forms in
\cite{mokhov18:_class_hamil} and for one significant example which comes from
centroaffine geometry~\cite{ferapontov04:_hyper}.

While we will look for third-order operators in the canonical form
\eqref{eq:18}, we shall explain the reason for choosing the ansatz
\eqref{eq:37} for the first-order operators.

Indeed, computational experiments show that Dubrovin's canonical form
$\eta^{(2)}$ does not admit a local first-order HHO. Then, the form
\eqref{eq:37} is the only possibility in the class of Ferapontov operators: the
vectors that multiply $\partial_x^{-1}$ must be (generalized) commuting
symmetries of the quasilinear systems of PDEs in view of the Hamiltonian
property of $A_1$ \cite{F95:_nl_ho} (see \eqref{eq:12}). Now, first-order WDVV
systems are non-diagonalizable, as we will discuss in
Section~\ref{sec:projectivegeom}, and non-diagonalizable systems \emph{with a
  low number of components} have only two such symmetries\footnote{There is
  only experimental evidence of this fact for $n=3$, $4$, $5$ (E.V. Ferapontov,
  private communication).}, namely $\varphi_1 = u^i_x\pd{}{u^i}$ and
$\varphi_2 = (V^i)_x\pd{}{u^i}$, which correspond to $t$ and $x$ translational
symmetries.  This means that the Hamiltonian property for an operator $A_1$ of
the form \eqref{eq:37} is equivalent to the conditions~\eqref{eq:64},
\eqref{eq:65} and
\begin{equation}
 \label{eq:88}
  \begin{split}
    R^{ij}_{kl} =& \alpha\big( V^i_k V^j_l - V^i_l V^j_k \big)
    \\
    &+\beta\big( V^i_k \delta^j_l - V^j_k \delta^i_l - V^i_l \delta^j_k +
    V^j_l \delta^i_k \big)
    +\gamma(\delta^i_k\delta^j_l - \delta^i_l\delta^j_k)
  \end{split}
\end{equation}
(obviously, the above two symmetries commute).  So, finding operators
\eqref{eq:37} amounts at finding the metric $g^{ij}$ and the three constants
$\alpha$, $\beta$, $\gamma$.

To this end, we recall a theorem in \cite{bogoyavlenskij96:_neces_hamil} that
states that, for non-diagonaliz\-able hydrodynamic-type systems in $n=3$ unknown
functions, the metric of a first-order Hamiltonian operator for the system
shall be proportional to a contraction of the square of the Haantjes tensor:
\begin{equation}
  \label{eq:3}
  g_{ij} = f\, H^\alpha_{i\beta}H^{\beta}_{j\alpha}, \qquad f=f(\mathbf{u}).
\end{equation}
See \cite[eq.\ 2.2 and 2.4]{bogoyavlenskij96:_neces_hamil} for a coordinate
expression of the Nijenhuis and the Haantjes tensors.

Summarizing, in the case $n=3$ there are one unknown function and three unknown
constants to be determined in order to find a first-order operator. We will use
the above equations in order to determine, by computer algebra, the first-order
operator $A_1$ for first-order WDVV systems.

\subsection{B.A. Dubrovin's normal forms of
  \texorpdfstring{$\eta$}{eta}}

We start with a result that follows from known calculations.

\begin{theorem}\label{th:dubeta1}
  Each first-order WDVV system in the orbit of $\eta^{(1)}$ under the action of
  the transformations~\eqref{eq:32} is endowed by a compatible pair of
  \emph{local} homogeneous Hamiltonian operators $A_1$ as in~\eqref{eq:17} and
  $A_2$ as in \eqref{eq:18}.
\end{theorem}
\begin{proof}
  The existence of a bi-Hamiltonian pair as in the statement was proved in
  \cite{FGMN97} (see also the Introduction). The fact that the bi-Hamiltonian
  pair propagates to each element of the orbit is a direct consequence of
  Theorem~\ref{th:invariance}.
\end{proof}

A completely new result holds for Dubrovin's second canonical form
\eqref{eq:10}. We stress that we could obtain this result only for recent
developments in the calculations of Schouten brackets for weakly nonlocal
operators \cite{casati19:_hamil,m.20:_weakl_poiss}.

\begin{theorem}\label{th:dubeta2}
  Each first-order WDVV system in the orbit of $\eta^{(2)}$ under the action of
  the transformations~\eqref{eq:32} is endowed by a compatible pair of
  homogeneous Hamiltonian operators $A_1$, a \emph{nonlocal} operator as
  in~\eqref{eq:37}, and $A_2$ as in \eqref{eq:18}.
\end{theorem}
\begin{proof}
  If $\eta_{11}\neq 0$, then we have the canonical form
\begin{equation}
  \label{eq:24}
  \eta =
  \begin{pmatrix}
    \mu & 0 & 1
    \\
    0 & 1 & 0
    \\
    1 & 0 & 0
  \end{pmatrix}
\end{equation}
with $\mu\neq 0$, to which it corresponds the equation
\begin{equation}
  \label{eq:25}
\mu f_{ttt}f_{xxt} - f_{ttt} + (f_{xxt})^2 - f_{xxx}f_{xtt} - \mu(f_{xtt})^2 = 0.
\end{equation}
The Algorithm yields the quasilinear first-order system
\begin{equation}
  \label{eq:26}
  \begin{split}
    &a_t = b_x,
    \\
    &b_t = c_x,
    \\
    &c_t = \left(\frac{ac - b^2 + \mu c^2}{\mu b - 1}\right)_x,
  \end{split}
\end{equation}
where we used the notation $a=f_{xxx}$, $b=f_{xxt}$,
$c=f_{xtt}$ for the sake of simplicity. We will adopt the same notation
throughout the rest of the Section.

The above system admits a third-order HHO $A_2$ which is completely
determined by the metric
\begin{equation}
    h_{ij}=
    \begin{pmatrix}
      b(\mu b -2) & (a+\mu c)(1-\mu b) & (\mu b-1)^2 \\
      (a+\mu c)(1-\mu b) & \mu(a+\mu c)^2+1 & \mu(a+\mu c)(1-\mu b) \\
      (\mu b-1)^2 & \mu(a+\mu c)(1-\mu b) & \mu(\mu b-1)^2%
    \end{pmatrix},
\end{equation}
and has the following form:
\begin{equation}\label{eq:27.1}
  A_2=
  \begin{pmatrix}
    -\mu\partial_x^3 & 0 & \partial_x^3 \\
    0 & \partial_x^3 & \partial_x^2 \frac{a+\mu c}{\mu b-1}\partial_x \\
    \partial_x^3 & \partial_x \frac{a + \mu c}{\mu b-1}\partial_x^2 & \frac{1}{2}(\partial_x^2 K \partial_x + \partial_x K \partial_x^2) \\
  \end{pmatrix},
\end{equation}
where 
\begin{equation}
K=\frac{(a+\mu c)^2+b(2-\mu b)}{(\mu b-1)^2}.\label{eq:62}
\end{equation}
We stress that the operator lies in the projective class $g^{(3)}$, according
to the classification in \cite{FPV16,FPV14}. Systems that possess a third-order
Hamiltonian operator have a non-local Hamiltonian that is specified
in~\cite{FPV17:_system_cl}.

Using the results from \cite{bogoyavlenskij96:_neces_hamil} (see the beginning
of the Section) we find that the system~\eqref{eq:26} has the first-order
operator of Ferapontov type $A_1$ defined by the metric

\footnotesize
\begin{equation}
  g^{ij}= 
  \begin{pmatrix}
    b^2\mu^2-a^2\mu-2b\mu-3 & a-ab\mu+bc\mu^2-c\mu & 2b-b^2\mu+c^2\mu^2 \cr 
    a-ab\mu+bc\mu^2-c\mu & 2b-b^2\mu+c^2\mu^2 & \frac{c( ac\mu^2 -2b^2\mu^2 +4b\mu +c^2\mu^3 -3)}{b\mu-1} \cr 
   2b-b^2\mu+c^2\mu^2  & \frac{c( ac\mu^2 -2b^2\mu^2 +4b\mu +c^2\mu^3 -3)}{b\mu-1} & \frac{\delta}{(b\mu-1)^2}
  \end{pmatrix},
\end{equation}
\normalsize where
\begin{multline*}
\delta=
a^2c^2\mu^2-2ab^2c\mu^2+4abc\mu+2ac^3\mu^3-4ac+b^4\mu^2-4b^3\mu-3b^2c^2\mu^3
\\
+4b^2+6bc^2\mu^2+c^4\mu^4-5c^2\mu
\end{multline*}
and the values of constants from (\ref{eq:37}) are
$\alpha=-\mu^2, \beta=0,\gamma=\mu$.

Using the results from \cite{CLV19} and the module developed in
\cite{m.20:_weakl_poiss} of the software package CDE we are able to prove the
compatibility of $A_1$ and $A_2$: the Schouten bracket $[A_1,A_2]$ vanish.

The fact that the bi-Hamiltonian pair propagates to each element of the orbit
is a direct consequence of Theorem~\ref{th:invariance}.
\end{proof}

\subsection{O.I. Mokhov and N.A. Pavlenko's normal forms of
  \texorpdfstring{$\eta$}{eta}}
\label{sec:o.i.-mokhov-n.a}

In a recent paper \cite{mokhov18:_class_hamil} Mokhov and Pavlenko classified
the WDVV equations \eqref{eq:5}, without the requirement of quasihomogeneity of
the solutions (item 3 in the Introduction), under transformations of the type
$T_1$ \eqref{eq:34}.  They came to the following four canonical forms:
\begin{fleqn}
  \begin{equation}
    \eta^1=
    \begin{pmatrix}
      0 & 0 & 1 \\
      0 & \lambda & 0 \\
      1 & 0 & \mu%
    \end{pmatrix},
    \ \lambda^2=1; \quad
    \eta^3=
    \begin{pmatrix}
      1 & 0 & 1 \\
      0 & 0 & 1 \\
      1 & 1 & 0%
    \end{pmatrix};
  \end{equation}
\begin{equation}
  \eta^2=
  \begin{pmatrix}
    1 & 0 & 1 \\
    0 & \lambda & 0 \\
    1 & 0 & \mu%
  \end{pmatrix},
  \ \lambda^2=1; \quad
  \eta^4=
  \begin{pmatrix}
    1 & 0 & 0 \\
    0 & \lambda & 0 \\
    0 & 0 & \mu %
  \end{pmatrix},\ \lambda^2=1, \ \mu^2=1.
\end{equation}
\end{fleqn}
According to Theorems~\ref{th:dubeta1} and \ref{th:dubeta2}, the above
examples can be rewritten as bi-Hamiltonian first-order WDVV systems.
Let us provide the first-order and third-order HHOs in all the above cases.

\paragraph{The case $\eta^1$.}
In this case in \cite{mokhov18:_class_hamil} it is provided a first-order local
HHO $A_1^1$. We look for a third-order one. The WDVV equations in first-order
form are of the type \eqref{eq:50}, where the function $\phi$ is defined by
\begin{equation}\label{eq:29}
  \phi=b^2-ac-\lambda\mu b +\mu^2.
\end{equation}
It is easy to show that the equations \eqref{V} have the unique solution
$h_{ij}$
\begin{equation}
    h_{ij}=
    \begin{pmatrix}
      -2b+\lambda\mu & a & 1 \\
      a & 1 & 0 \\
      1 & 0 & 0 %
    \end{pmatrix}.
\end{equation}
which is of type $g^{(5)}$ (as~\eqref{eq:6}) according to the classification
in~\cite{FPV14}. We will omit the corresponding operator $A_2^1$, as it can be
easily reconstructed.

The Schouten bracket $[A^1_1,A^1_2]$ turns out to be zero and thus the above
operators are compatible; the computation has been performed by means of the
Reduce package CDE, see \cite{KVV17,vitolo17:_hamil_pdes}.

\paragraph{The case $\eta^2$.}
For $\eta^2$ we have the function $\phi(a,b,c)$:
\begin{equation}
  \phi(a,b,c)=\frac{b^2-ac-\lambda c^2-\mu b^2+\mu ac-\mu\lambda b+\mu^2}{1-\lambda b},
\end{equation}
for an arbitrary real constant $\mu$.

In this case, the WDVV system in first-order form has a third-order HHO
which is completely determined by the metric of type $g^{(3)}$
\cite{FPV16,FPV14}:
\begin{equation}
  h_{ij}=
  \begin{pmatrix}
    \mu-b(2\lambda-b) & \frac{(\lambda-b)(\mu a-\lambda c -a)}{\mu-1} & \frac{-(\lambda-b)^2\lambda}{\mu-1} \\
    \frac{(\lambda-b)(\mu a-\lambda c -a)}{\mu-1} & \frac{(\mu - 1)^2(\lambda + a^2) - 2(\mu - 1)\lambda ac +  c^2}{(\mu-1)^2} &  \frac{-(\mu a-\lambda c-a)(\lambda -b)\lambda}{(\mu-1)^2} \\
   \frac{-(\lambda-b)^2\lambda}{\mu-1} & \frac{-(\mu a-\lambda c-a)(\lambda -b)\lambda}{(\mu-1)^2} & \frac{(\lambda-b)^2}{(\mu-1)^2} \\
  \end{pmatrix},
\end{equation}
Note that $\mu \neq 1$, as $\det \eta^2\neq 0$.

The WDVV system in first-order form also admits a non-local Hamiltonian
operator of Ferapontov type $A^2_1$ that is compatible with $A^2_2$. We present
the special case where $\mu=2$, as the general case has an expression that is
too big to be shown here. We have
\begin{equation}
  \label{eq:38}
  g^{ij}=\begin{pmatrix}
    2b\lambda-a^2\lambda -b^2-5  & a-ab\lambda-bc+c\lambda & 2b-b^2\lambda-c^2-2\lambda  \cr
    a-ab\lambda-bc+c\lambda & 2b-b^2\lambda-c^2-2\lambda & \frac{c(ac\lambda-2b^2\lambda+4b-c^2-\lambda)}{b-\lambda}  \cr
    2b-b^2\lambda-c^2-2\lambda & \frac{c(ac\lambda-2b^2\lambda+4b-c^2-\lambda)}{b-\lambda} &
    \frac{\delta}{(b-\lambda)^2},
  \end{pmatrix}
\end{equation}
where
$$ 
\delta=  \lambda(b^2(4b-3c^2)-2ac(2b-c^2)+8b-c^2) - (4b-c^2)(2b-c^2) - 4 - (ac-b^2)^2
$$
and the value of the constants in the non-local part is
$\alpha=1,\beta=0,\gamma=\lambda$.

\paragraph{The case $\eta^3$.}
In this case we have
\begin{equation}
  \phi(a,b,c)=\frac{b^2-ac+bc-2b+1}{a}.
\end{equation}
There exists a unique third-order HHO $A^3_2$, it is generated by the following
metric of type $g^{(4)}$ \cite{FPV16,FPV14}:
\begin{equation}
  h_{ij}=
  \begin{pmatrix}
    (1-b-c)(b+c-3) & a(b+c-2)+1 & a(b+c-2) \\
    a(b+c-2) +1 & -a^2 &  -a^2 \\
    a(b+c-2) & -a^2 & -a^2 \\
  \end{pmatrix}.
\end{equation}

We also have an operator $A_1^3$ of Ferapontov type (\ref{eq:37}). It is
determined by the metric
\begin{equation}
  g^{ij}=
  \begin{pmatrix}
   - a^2 - 2ab & 2a - ab - ac - b^2 & 4b - b^2 - 2bc - 3  \cr 
  2a - ab - ac - b^2 &   4b - b^2 - 2bc - 3 & \frac{4ac - ac^2 - b^3 - b^2c + 2b^2 - b}{a}  \cr 
  4b - b^2 - 2bc - 3 & \frac{4ac - ac^2 - b^3 - b^2c + 2b^2 - b}{a} &  \frac{\delta}{a}
  \end{pmatrix},
\end{equation}
where $\delta=ac^2 - 4ac - 2b^2c + 4b^2 - 2bc^2 + 8bc - 8b - 2c + 4$. The
nonlocal part of the operator is determined by $\alpha=0,\beta=1,\gamma=1$.

Furthermore, the Schouten bracket vanishes: $[A^3_1,A^3_2]=0$.

\paragraph{The case $\eta^4$.}
In this case the function $\phi(a,b,c)$
\begin{equation}
  \phi(a,b,c)=\frac{\mu b^2-\mu ac+\lambda c^2-1}{\lambda b}.
\end{equation}
There exists a third-order HHO which is in the class $g^{(3)}$
\cite{FPV16,FPV14}. It is determined by the metric
\begin{equation}
  h_{ij}=
  \begin{pmatrix}
    b^2+\mu & b\mu(\lambda c-\mu a) & -\mu\lambda b^2 \\
    b\mu(\lambda c-\mu a) & \lambda+a^2-\lambda c(2\mu a-\lambda c) &
    \lambda b(\mu a-\lambda c) \\
   -\mu\lambda b^2 & \lambda b(\mu a-\lambda c)) & b^2 \\
  \end{pmatrix}.
\end{equation}
Below we provide the metric $g^{ij}$ that defines $A_1^4$:
\begin{equation}
  g^{ij}= \frac{\lambda}{\mu}
  \begin{pmatrix}
    -a^2\mu-b^2\lambda-4\mu\lambda & -b( a\mu+c\lambda) &
    -b^2\mu - c^2\lambda - 1 \cr 
    -b( a\mu+c\lambda) & -b^2\mu - c^2\lambda - 1 &
    \frac{c( ac\mu -2b^2\mu -c^2\lambda + 1)}{b} \cr 
    -b^2\mu - c^2\lambda - 1 & \frac{c(ac\mu -2b^2\mu -c^2\lambda + 1)}{b}
    & \frac{\delta}{b^2}
  \end{pmatrix},
\end{equation}
where $\delta= 2ab^2c\lambda-a^2c^2\lambda +2ac^3\mu -2ac\mu\lambda
-b^4\lambda -3b^2c^2\mu -2b^2\mu\lambda -c^4\lambda +2c^2 - \lambda$
and the values of constants in (\ref{eq:37}) are
$\alpha=\mu,\beta=0, \gamma=\lambda$. Also in this case $[A^4_1,A^4_2]=0$ for
any $\lambda,\mu = \pm 1$.

\subsection{An example: equation of flat centroaffine metrics}

There is another significant example provided in \cite{D96} in the case $N=3$:
\begin{equation}\label{eq:28}
    \eta=
    \begin{pmatrix}
      1 & 0 & 0 \\
      0 & 0 & 1 \\
      0 & 1 & 0%
    \end{pmatrix}.
\end{equation}
This particular example has an interesting geometric interpretation as the
equation of flat centroaffine metrics for surfaces in
$\mathbb{R}^3$~\cite{ferapontov04:_hyper}. The WDVV equation takes the form
\begin{equation}\label{eq:1}
f_{xxx}f_{yyy} - f_{xxy}f_{xyy} = 1,
\end{equation}
and the system in first-order form reads:
\begin{equation}\label{eq:27}
  \begin{split}
    &a_t = b_x, \\
    &b_t = c_x, \\
    &c_t = \left(\frac{bc+1}{a}\right)_x.
  \end{split}%
\end{equation}
The above system admits a third-order HHO of type $g^{(4)}$ which is again
completely determined by a metric $h_{ij}$:
\begin{equation}
    h_{ij}=
    \begin{pmatrix}
      c^2 & -1 & -ac \\
      -1 & 0 & 0 \\
      -ac & 0 & a^2%
    \end{pmatrix}.
\end{equation}
and has the form:
\begin{equation}\label{eq:28.1}
  A_2= 
  \begin{pmatrix}
    0 & -\partial_x^3 & 0 \\
    -\partial_x^3 & 0 & -\partial_x^2\frac{c}{a}\partial_x \\
    0 & -\partial_x\frac{c}{a}\partial_x^2 &
    \frac{1}{2}(\partial_x^2\frac{1}{a^2}\partial_x
    + \partial_x\frac{1}{a^2}\partial_x^2) \\
  \end{pmatrix}.
\end{equation}

For this case we also have a first-order nonlocal operator $A_1$ which is
determined by metric $g^{ij}$:
\begin{equation}
  g^{ij}= 
  \begin{pmatrix}
   2ab & ac+b^2 & 2bc+3 \cr
   ac+b^2 & 2bc+3 & \frac{ac^2+b^2c+b}{a} \cr
   2bc+3 & \frac{ac^2+b^2c+b}{a} & \frac{2c(bc+1)}{a}
  \end{pmatrix},
\end{equation}
and its tail vector $\alpha=\gamma=0, \beta=-1$. The operator is compatible
with $A_2$: $[A_1,A_2]=0$. This turns \emph{the equation of flat centroaffine
metrics, in the form of a first-order quasilinear system of PDEs \eqref{eq:27},
into a  bi-Hamiltonian system.}

Example 3 in \cite{FPV16} is just another form of~\eqref{eq:1}, after
transforming $\eta^{(2)}$ into the identity matrix ($\mu=1$). Also this case is
bi-Hamiltonian by means of the third-order HHO in Example 3 and the first-order
nonlocal operator presented in \cite{m.20:_weakl_poiss}.

\begin{remark}
  Ferapontov operators of type \eqref{eq:37} can be transformed into local
  operators by a reciprocal transformation if~$\alpha=\gamma=0$
  \cite{ferapontov03:_recip_hamil}. The fact that in the above case the
  first-order operator is localizable does not contradict the fact that,
  according to \cite{mokhov18:_class_hamil}, it is \emph{not} localizable by
  equivalence transformations of WDVV. Indeed, in \cite{FM96:_equat_hamil} a
  transformation between \eqref{eq:13} and the above \eqref{eq:1} which can
  also be read as a reciprocal transformation between the first-order systems
  is provided. It should not be a difficult exercise to show that such a
  transformation brings the local first-order HHO of \eqref{eq:13} into the
  above non-local first-order HHO.
\end{remark}

\subsection{Remarks on the generic case}
\label{sec:remarks-generic-case}

A computation in the generic case is possible. In particular, if $N=3$ a
generic choice of the matrix $\eta$ leads to a first-order system of the
form~\eqref{eq:50}. We can prove (by computer) that if $\eta_{11}\neq 0$ the
generic metric of the third-order operator is of type $g^{(3)}$, while if
$\eta_{11}=0$ it is of type $g^{(4)}$ (see \cite{FPV14}).  Such computations
show that, at least in the case $N=3$, the third-order operators exist
independently, being the first-order system generated by the algorithm or as a
relation between the coordinates f involving the WDVV equation. This fact
suggests that a natural framework for characterizing the existence of a
third-order HHO might be that of \cite{Agafonov_2019}, where \emph{all}
conserved densities of the first-order system (in our case 5) are used to
identify the system with a projective variety.

It is interesting to observe that our examples do not cover the whole range of
Monge metrics that are classified in \cite{FPV14}, as the most generic types
$g^{(1)}$ and $g^{(2)}$. However, strictly speaking an occurrence of a metric
of the latter classes in first-order WDVV systems cannot be excluded.

It is interesting to observe that, when $N=3$, the first-order WDVV systems
produced by the Algorithm are always of the \emph{reducible}
type~\eqref{eq:48}, according with the terminology in
\cite{agafonov98:_linear}. In the same paper it is proved that linearly
degenerate reduced systems are semi-Hamiltonian. In our case we proved that
such systems admit a third-order HHO, and hence they are linearly degenerate
(see \cite{FPV17:_system_cl} and Section~\ref{sec:projectivegeom}), but we
prove that the systems are bi-Hamiltonian (which, of course, imply
semi-Hamiltonianity).

\section{WDVV equations in higher dimensions}
\label{sec:wdvv-equations-n=4}

We still have no results of invariance of the bi-Hamiltonian structure for
higher dimensions. However, there are strong indications that third-order HHOs
exist for both normal forms $\eta^{(1)}$ and $\eta^{(2)}$ and all values of
$N$, due to their existence for $N=4$ and $N=5$. In one case, the third-order
HHO is paired with a first-order local HHO and we have a bi-Hamiltonian
pair. We present the main results in this section.

We shall remark that at the moment there exists only one first-order operator
for $N=4$ and $\eta=\eta^{(1)}$: it was found in \cite{ferapontov96:_hamil}.
For $N=4$ and $\eta=\eta^{(2)}$ and $N=5$ we cannot exhibit any first-order
operator (local or non-local) because the properties that we used in the case
$N=3$ \cite{bogoyavlenskij96:_neces_hamil} do not hold in higher dimensions.
It is however reasonable to \emph{conjecture} that we will continue to have
first-order operators, local in the case $\eta=\eta^{(1)}$ and
nonlocal of type \eqref{eq:37} in the case $\eta=\eta^{(2)}$, and that they
will form bi-Hamiltonian pairs with the corresponding (conjectural) third-order
operators.

\subsection{The case \texorpdfstring{$N=4$}{N=4}}

In the case $N=4$ for long time only a first-order local HHO was known
\cite{ferapontov96:_hamil} in the case $\eta^{(1)}$ \eqref{eq:14}. A few
years ago, a third-order HHO compatible with the known first-order HHO was
found by a complicated procedure \cite{PV15}, later simplified in
\cite{FPV17:_system_cl} by means of equations~\eqref{V}.

Here we will present a third-order HHO for Dubrovin's normal form $\eta^{(2)}$.
This particular form of $\eta$ generates the following WDVV system (here
$(t,x,y,z) = (t^1,t^2,t^3,t^4)$):

\small
\begin{equation}\label{eq:22}
  \begin{split}
    &\mu f_{yyz}f_{zzz} + 2f_{yyz}f_{xyz} - f_{yyy}f_{xzz} - f_{xyy}f_{yzz} - \mu f_{yzz}^2=0,\\
    & f_{xxy}f_{yzz} - f_{xxz}f_{yyz} - \mu f_{zzz}f_{xyz} + f_{zzz} + f_{xyy}f_{xzz} + \mu f_{xzz}f_{yzz}   - f_{xyz}^2=0,\\
    & f_{xxy}f_{yyz} - f_{xxz}f_{yyy} + \mu f_{yyz}f_{xzz} - \mu f_{xyz}f_{yzz} + f_{yzz}=0,\\
    & f_{xxy}f_{xzz} - \mu f_{xxz}f_{zzz} - 2f_{xxz}f_{xyz} + f_{xxx}f_{yzz} + \mu f_{xzz}^2=0,\\
    & f_{xxz}f_{xyy} + \mu f_{xxz}f_{yzz} - f_{yyz}f_{xxx} - \mu f_{xzz}f_{xyz} + f_{xzz}=0,\\
    & f_{xxy}f_{xyy} + \mu f_{xxz}f_{yyz} - f_{xxx}f_{yyy} - \mu f_{xyz}^2 +
    2f_{xyz}=0. %
  \end{split}
\end{equation}
\normalsize The above overdetermined system of non-linear PDEs can be rewritten
in the form of two commuting hydrodynamic-type system using the
identifications
$a=f_{xxx}, b=f_{xxy}, c=f_{xxz}, d=f_{xyy}, e=f_{xyz}, f=f_{xzz}$. One of the
systems is:

\small
\begin{equation}
  \begin{split}
    a_y &= b_x, \qquad b_y = d_x, \qquad c_y = e_x, \\
    d_y &=\left(\frac{bcdf\mu^2-bce^2\mu^2-c^2de\mu^2+ae^3\mu^2+b^2cd\mu -}{\delta}\right. \\
    &+ \left. \frac{-bade\mu+2bce\mu+c^2d\mu-3ae^2\mu+bad+cf\mu+2ae}{\delta}\right)_x, \\
    e_y &=\left(\frac{-c^2fe\mu^2+af^2e\mu^2+bafe\mu+c^3d\mu - }{\delta}\right. \\
    &+ \left. \frac{-cadf\mu-cae^2\mu+c^2f\mu-af^2\mu-baf+2cae}{\delta}\right)_x, \\
    f_y &= \left(\frac{-c^2e^2\mu^2+afe^2\mu^2+bc^2d\mu-cade\mu+2c^2e\mu-2afe\mu+cad+af}{\delta}\right)_x,\\
  \end{split}
\end{equation}
\normalsize where $\delta=-c^3\mu^2+caf\mu^2+bca\mu-a^2e\mu+a^2$. As an
example, the right-hand side of the fourth equation comes from the
compatibility condition $d_y=(f_{yyy})_x$ and the expression of $f_{yyy}$ in
terms of $a$, \dots, $f$ through the system~\eqref{eq:22}.

The above quasilinear system of first-order PDEs admits a unique third-order
HHO determined by the following Monge metric $h_{ij}$ through the
equations~\eqref{V} (only entries with $i<j$ are shown): \allowdisplaybreaks[3]

\small
\begin{alignat*}{2}
  &h_{11}= d^2
  &\ & h_{12}=e^2\mu - 2e \\
  &h_{13}= 2d( - e\mu + 1)
  && h_{14}= - ad + ce\mu - c \\
  &h_{15}= \mu(b/\mu - be + cd - ef\mu + f)
  && h_{16}= e^2\mu^2 - 2e\mu + 1 \\
  &h_{22}= 2c(e\mu - 1)
  && h_{23}= - be\mu + b - cd\mu - ef\mu^2 + f\mu \\
  &h_{24}= c^2\mu
  && h_{25}= - ae\mu + a - bc\mu - cf\mu^2 \\
  &h_{26}= 2c\mu(e\mu - 1)
  && h_{33}= 2\mu(bd + df\mu + e^2\mu/2 - e + 1/\mu) \\
  &h_{34}= ae\mu - a - bc\mu - cf\mu^2
  && h_{35}= \mu(ad + b^2 + 2bf\mu - ce\mu + c + f^2\mu^2) \\
  &h_{36}= \mu(b - be\mu - cd\mu - ef\mu^2 + f\mu)
  && h_{44}= a^2 \\
  &h_{45}= - 2ac\mu
  && h_{46}= c^2\mu^2 \\
  &h_{55}= \mu(2ab + 2af\mu + c^2\mu)
  && h_{56}= \mu( - ae\mu + a - bc\mu - cf\mu^2) \\
  &h_{66}= 2c\mu^2(e\mu - 1) &&\\
\end{alignat*}
\normalsize
\begin{remark}
  We have $2$ commuting quasilinear first-order systems of PDEs for each choice
  of $\eta$, and it is not automatically true that if one of them is
  bi-Hamiltonian the other will be bi-Hamiltonian with respect to the
  \emph{same operators}. However, this is true in the case $\eta^{(1)}$ and
  $N=4$ \cite{PV15}. More generally, it is known
  \cite{tsarev91:_hamil,tsarev85:_poiss_hamil} that if a diagonalizable
  quasilinear first-order system of PDEs is Hamiltonian with respect to a
  first-order HHO, then other commuting diagonalizable systems will be
  Hamiltonian with respect to the same operator. Even if this statement carries
  on to our non-diagonalizable first-order WDVV systems (see
  Section~\ref{sec:projectivegeom}), at the moment we can only
  \emph{conjecture} that this extends to the compatible third-order operators.
\end{remark}

\subsection{The case \texorpdfstring{$N=5$}{N=5}}
\label{sec:wdvv-equation-n=5}

This case is completely open: no Hamiltonian formulation was known until now.
We have been able to find one new third-order Hamiltonian operator for the
normal form $\eta^{(1)}$ \eqref{eq:14}. The first-order WDVV systems are
$10$-component systems; one of them (found using $t^2$ and $t^3$ in the
Algorithm) admits one third-order HHO (up to a constant factor) that is defined
by the following Monge metric $h_{ij}$ (only nonzero entries are shown):

\small
\begin{alignat*}{2}
  & h_{11}= -u^7  &\ & h_{12}= -2u^6u^7 \\
  &h_{13}= (u^6)^2+2u^9 && h_{14}= -2u^7 \\
  &h_{15}= -(u^6)^2 && h_{16}= u^2u^7-u^3u^6+u^5u^6+u^8 \\
  &h_{17}= u^1u^7+u^2u^6+u^4  && h_{18}= -2u^6 \\
  &h_{19}= -u^3 && h_{110}= -1 \\
  &h_{22}= 2u^3u^7-2u^5u^7-(u^6)^2+2u^9& & h_{23}= -u^2u^7-u^3u^6+u^5u^6+u^8 \\
  &h_{24}=  -2u^6 && h_{25}= u^2u^7+u^3u^6-u^5u^6+u^8 \\
  &h_{26}= u^1u^7+u^2u^6+(u^3-u^5)^2+u^4 & & h_{27}= u^1u^6-u^2u^3+u^2u^5 \\
  &h_{28}= u^3-2u^5 && h_{29}= -u^2 \\
  &h_{33}= 2u^2u^6+2u^4  && h_{34}= -u^3 \\
  &h_{35}= -2u^2u^6  && h_{36}= -u^1u^6-u^2u^3+u^2u^5 \\
  &h_{37}= (u^2)^2  && h_{38}= -2u^2 \\
  &h_{39}= -u^1  && h_{44}= -2 \\
  &h_{46}= u^2  && h_{47}= u^1 \\
  &h_{55}= 2u^2u^6  && h_{56}= u^1u^6+u^2u^3-u^2u^5 \\
  &h_{57}= -(u^2)^2  && h_{58}= u^2 \\
  &h_{66}= 2u^1u^3-2u^1u^5-(u^2)^2  && h_{67}= -2u^1u^2 \\
  &h_{68}= u^1  && h_{77}= -(u^1)^2 \\
\end{alignat*}
\normalsize\indent For the second normal form $\eta^{(2)}$ we were able to find
a third order operator as well, with the simplifying assumption $\mu=1$ (again,
we used $t^2$ and $t^3$ in the Algorithm). Indeed, the computation is too hard
for the servers that we can use if we allow $\mu$ to be a generic non-zero
constant.  The operator is defined by the following Monge metric $h_{ij}$
(again, we present only non-zero entries with $i\leq j$):

\footnotesize
\begin{alignat*}{2}
  &h_{1\,1}\!= (u^7)^2
  &\ & h_{1\,2}\!= 2u^7u^6 \\
  &h_{1\,3}\!= \!-\!(u^6)^2\!+\!(u^9)^2\!-\!2u^9
  && h_{1\,4}\!= \!-\!(2u^9\!-\!2)u^7 \\
  &h_{1\,5}\!= (u^6)^2
  && h_{1\,6}\!= \!-\!u^2u^7\!+\!u^3u^6\!-\!u^5u^6\!+\!u^8u^9\!-\!u^8\\
  &h_{1\,7}\!= \!-\!u^1u^7\!-\!u^2u^6\!+\!u^4u^9\!-\!u^4
  && h_{1\,8}\!= \!-\!(2u^9\!-\!2)u^6 \\
  &h_{1\,9}\!=
  \!-\!u^{10}u^9\!-\!u^3u^9\!+\!u^4u^7\!+\!u^6u^8\!+\!u^{10}\!+\!u^3
  && h_{1\, 10}\!= (u^9)^2\!-\!2u^9\!+\!1 \\
  &h_{2\,2}\!= \!-\!2u^3u^7\!+\!2u^5u^7\!+\!(u^6)^2\!+\!(u^9)^2\!-\!2u^9
  && h_{2\,3}\!= u^2u^7\!+\!u^3u^6\!-\!u^5u^6\!+\!u^8u^9\!-\!u^8 \\
  &h_{24\,}\!= \!-\!2u^6u^9\!-\!2u^7u^8\!+\!2u^6
  && h_{2\,5}\!= \!-\!u^2u^7\!-\!u^3u^6\!+\!u^5u^6\!+\!u^8u^9\!-\!u^8 \\
  &h_{2\,6}\!= \!-\!u^1u^7\!-\!u^2u^6\!-\!(u^3)^2\!+\!2u^3u^5\!+
  && h_{2\,7}\!= \!-\!u^1u^6\!+\!u^2u^3\!-\!u^2u^5\!+\!u^4u^8 \\
  &\qquad + \!u^4u^9\!-\!(u^5)^2\!+\!(u^8)^2\!-\!u^4 && \\
  &h_{2\,8}\!= u^9(u^3 \!+\!u^{10}\!-\!2u^5)\!+\!u^4u^7\!-\!u^6u^8\!
  && h_{2\,9}\!= u^8(u^5\!-\!u^{10}\!-\!2u^3)\!-\!u^2u^9\!+\!u^4u^6\!+\!u^2 \\
  &\qquad+\!u^{10}\!-\!u^3\!+\!2u^5 &&\\
  &h_{2\,10}\!=(2u^9\!-\!2)u^8
  && h_{3\,3}\!= \!-\!2u^2u^6\!+\!2u^4u^9\!-\!2u^4\\
  &h_{3\,4}\!=
  \!-\!u^{10}u^9\!-\!u^3u^9\!-\!u^4u^7\!+\!u^6u^8\!+\!u^{10}\!+\!u^3
  && h_{3\,5}\!=  2u^6u^2\\
  &h_{3\,6}\!= u^1u^6\!+\!u^2u^3\!-\!u^2u^5\!+\!u^4u^8
  && h_{3\,7}\!= \!-\!(u^2)^2\!+\!(u^4)^2 \\
  &h_{3\,8}\!= \!-\!2u^2u^9\!-\!2u^4u^6\!+\!2u^2
  && h_{3\,9}\!= \!-\!u^1u^9\!-\!u^{10}u^4\!+\!u^2u^8\!-\!u^3u^4\!+\!u^1 \\
  &h_{3\,10}\!= (2u^9\!-\!2)u^4
  && h_{4\,4}\!=2u^{10}u^7\!+\!2u^3u^7\!+\!(u^6)^2\!+\!(u^9)^2\!-\!2u^9\!+\!2 \\
  &h_{4\,5}\!= \!-\!2u^8u^6
  && h_{4\,6}\!=\!-\!u^{10}u^8\!+\!u^2u^9\!-\!2u^3u^8\!-\!u^4u^6\!+\!u^5u^8\!-\!u^2\\
  &h_{4\,7}\!= u^1u^9\!-\!u^{10}u^4\!+\!u^2u^8\!-\!u^3u^4\!-\!u^1
  && h_{4\,8}\!= 2u^{10}u^6\!+\!u^2u^7\!+\!u^3u^6\!+\!u^5u^6\!+\!u^8u^9\!-\!u^8 \\
  &h_{4\,9}\!= u^1u^7\!+\!(u^{10})^2\!+\!2u^{10}u^3\!+\!u^2u^6\!+
  && h_{4\,10}\!= \!-\!u^{10}u^9\!-\!u^3u^9\!-\!u^4u^7\!-\!u^6u^8\!+\!u^{10}\!+\!u^3 \\
  &\qquad + \!(u^3)^2\!-\!u^4u^9\!-\!(u^8)^2\!+\!u^4 &&\\
  &h_{5\,5}\!= \!-\!2u^2u^6\!+\!(u^8)^2
  && h_{5\,6}\!=\!-\!u^1u^6\!-\!u^2u^3\!+\!u^2u^5\!+\!u^4u^8 \\
  &h_{5\,7}\!= (u^2)^2
  && h_{5\,8}\!= \!-\!u^{10}u^8\!+\!u^2u^9\!+\!u^4u^6\!-\!u^5u^8\!-\!u^2 \\
  &h_{5\,9}\!= \!-\!2u^8u^2
  && h_{5\,10}\!= (u^8)^2\\
  &h_{6\,6}\!= \!-\!2u^1u^3\!+\!2u^1u^5\!+\!(u^2)^2\!+\!(u^4)^2
  && h_{6\,7}\!=2u^2u^1 \\
  &h_{6\,8}\!= u^1u^9\!-\!u^{10}u^4\!-\!u^2u^8\!+\!u^3u^4\!-\!2u^4u^5\!-\!u^1
  && h_{6\,9}\!= \!-\!2u^1u^8\!-\!2u^2u^4 \\
  &h_{6\,10}\!= 2u^8u^4
  && h_{7\,7}\!= (u^2)^2 \\
  &h_{7\,8}\!= \!-\!2u^4u^2
  && h_{7\,9}\!= \!-\!2u^4u^1 \\
  &h_{7\,10}\!= (u^4)^2
  && h_{8\,8}\!= (u^{10})^2\!+\!2u^{10}u^5\!+\!2u^2u^6\!-\!2u^4u^9\!+\!(u^5)^2\!+\!2u^4 \\
  &h_{8\,9}\!= u^1u^6\!+\!2u^{10}u^2\!+\!u^2u^3\!+\!u^2u^5\!+\!u^4u^8
  && h_{8\,10}\!= \!-\!u^{10}u^8\!-\!u^2u^9\!-\!u^4u^6\!-\!u^5u^8\!+\!u^2 \\
  &h_{9\,9}\!= 2u^1u^{10}\!+\!2u^1u^3\!+\!(u^2)^2\!+\!(u^4)^2
  && h_{9\,10}\!= \!-\!u^1u^9\!-\!u^{10}u^4\!-\!u^2u^8\!-\!u^3u^4\!+\!u^1 \\
  &h_{10\,10}\!= 2u^4u^9\!+\!(u^8)^2\!-\!2u^4 &&\\
\end{alignat*}
\normalsize

\section{Conclusions: projective geometry of WDVV systems}
\label{sec:projectivegeom}

The results that we achieved so far have very interesting implications in terms
of projective geometry that we will discuss in this Section (see also the
Introduction).

In a recent paper \cite{FPV17:_system_cl} it was proved that any
hydrodynamic-type system that has a third-order homogeneous Hamiltonian
operator has a rich geometric structure. The results transfer to first-order
WDVV systems. Indeed, in the case $N=3$ we have:
\begin{enumerate}
\item There exists a \emph{quadratic line complex} of lines in the projective
  space $\mathbb{P}^3$ associated with the first-order WDVV system. The
  construction is a straightforward generalization of what we have written in
  the Introduction.
\item The first-order WDVV system defines a \emph{linear line congruence} in
  the projective space $\mathbb{P}^4$. This is a $3$-parameter family of lines
  in the projective space $\mathbb{P}^{4}$. Again, this is shown in the
  Introduction.
\item The first-order WDVV system is \emph{linearly degenerate} and belongs to
  the \emph{Temple class}.
\item The first-order WDVV system is \emph{non-diagonalizable}.
\item The first-order WDVV system admits a Hamiltonian and a momentum with
  respect to the third-order HHO; their expressions are local after a potential
  substitution $b^i_x = u^i$ and are given by explicit formulae
  \cite{FPV17:_system_cl}.
\item The first-order WDVV systems are equivalent to the system \eqref{eq:6}
  using a projective reciprocal transformation of the type~\eqref{recip}.
\end{enumerate}

Let $N=4$, $5$. Then, the above properties (with the exception of equivalence)
hold for the first-order WDVV systems that we considered in
Section~\ref{sec:wdvv-equations-n=4} for the normal forms $\eta^{(1)}$,
$\eta^{(2)}$.

Several considerations can be made in view of future research.

\paragraph{Conjecture: every WDVV system is associated with a quadratic line
  complex and a linear line congruence.} This is just a rephrasing of the
conjecture that every first-order WDVV system admits a local third-order
HHO. However, the Conjecture put in this way suggests that WDVV equations have
a projective-geometric interpretation that is yet to be uncovered.  The
generalization to the Oriented Associativity equations for $F$-manifolds is an
active research topic (see \emph{e.g.}  \cite{ABLR20:_semis_f}) and quadratic
line complexes occur also in that case \cite{m.v.19:_bi_hamil_orien_assoc} (see
also \cite{pavlov14:_orien}).  A future role of such objects besides the
bi-Hamiltonian structure that they provide is foreseeable.

\paragraph{Projective equivalence of WDVV systems.} The last statement in the
above list implies that, in the case $N=3$, if our WDVV systems admit a
third-order homogeneous operator they should in principle be all equivalent.
In practice, it is extremely difficult to solve the equations for the unknown
transformation. An explicit transformation between the WDVV
equations~\eqref{eq:13} and~\eqref{eq:1} in \cite{ferapontov96:_hamil} shows
the computational difficulties. It is preferable to recompute the operators for
any specific presentation of the WDVV system. Note that the
transformation~\eqref{recip} is \emph{not} an invariance transformation of WDVV
equations. In principle, the equivalence does no longer hold in higher
dimensions.

\paragraph{First-order operators and projective geometry.} First-order local or
non-local HHOs do not have a projective-geometric interpretation yet. However,
we expect that those that are \emph{compatible} with third-order HHOs will have
a projective-geometric role. It is an interesting remark the fact that the
metric of the first-order operators that we found have rational coefficients
(in upper indices!) and the denominator is always the square root of the
determinant of the Monge metric. Such a determinant is a perfect square, and
its zero locus is the K\"ummer surface of the underlying quadratic line
complex, see \cite{FPV16,FPV14}.

\paragraph{Projective geometry in the original WDVV setting.} The
projective-geometric structures that we found so far might be transported in
the `initial' setting of WDVV. A theoretical framework for the Hamiltonian
formalism for general PDEs has already been developed in
\cite{KerstenKrasilshchikVerbovetskyVitolo:HSGP}. In that framework,
variational bivectors for the WDVV equation in the non-evolutionary
form~\eqref{eq:13} have been found (although with an explicit dependence on the
independent variables) \cite{KKVV}. An understanding of the role of quadratic
line complexes and line congruences in the initial formulation might shed new
light on the relationship between WDVV equations and projective-geometric
invariants.

\bigskip

\textbf{Acknowledgements.}  R.V. would like to thank M. Casati,
E.V. Ferapontov, J.S. Krasil'shchik, P. Lorenzoni, M.V. Pavlov, A. Sergyeyev,
A.M. Verbovetsky for many scientific discussions throughout years of scientific
cooperation. Thanks are also due to M. Beccaria for useful suggestions.

Computational resources were supplied by the project ``e-Infrastruktura CZ''
(e-INFRA LM2018140) provided within the program Projects of Large Research,
Development and Innovations Infrastructures and by the Dept. of Mathematics and
Physics ``E. De Giorgi'' of the Universit\`a del Salento.

The research of JV was supported in part by the Specific Research grant
SGS/13/2020 of the Silesian University in Opava, Czechia.

The research of RV has been funded by the Dept. of Mathematics and Physics
``E. De Giorgi'' of the Universit\`a del Salento, Istituto Naz. di Fisica
Nucleare IS-CSN4 \emph{Mathematical Methods of Nonlinear Physics}, GNFM of
Istituto Nazionale di Alta Matematica.

\small

\providecommand{\cprime}{\/{\mathsurround=0pt$'$}}
  \providecommand*{\SortNoop}[1]{}

\end{document}